\documentclass[a4paper,11pt]{article}

\usepackage[hidelinks]{hyperref}
\usepackage[utf8]{inputenc}
\usepackage[T1]{fontenc}
\usepackage[small]{caption}
\usepackage{amsmath}
\usepackage{amsthm}
\usepackage{microtype}
\usepackage{fullpage}
\usepackage{amssymb}
\usepackage{cleveref}
\usepackage{centernot}
\usepackage{dsfont}

\usepackage{tikz}
\usetikzlibrary{arrows}

\theoremstyle{plain}
\newtheorem{theorem}{Theorem}
\theoremstyle{definition}
\newtheorem{definition}[theorem]{Definition}
\theoremstyle{remark}

\theoremstyle{plain}
\newtheorem{lemma}[theorem]{Lemma}

\newtheorem{claim}[theorem]{Claim}

\DeclareMathOperator*\Prob{\bf Pr}

\newcommand{\prn}[1]{\left(#1\right)}
\newcommand{\cprn}[1]{\!\left(#1\right)}
\newcommand{\sqbra}[1]{\left[#1\right]}
\newcommand{\csqbra}[1]{\!\left[#1\right]}
\newcommand{\sqprn}[1]{\left[#1\right)}

\newcommand{\abs}[1]{\left|#1\right|}

\newcommand{\brkts}[1]{\left\{#1\right\}}

\newcommand{\bool}{\brkts{0,1}}

\newcommand{\poly}[1]{\mathrm{poly}\prn{#1}}
\newcommand{\cpoly}[1]{\mathrm{poly}\cprn{#1}}

\newcommand{\cprq}[2]{\Prob_{#2}\csqbra{#1}}

\newcommand{\TV}{\mathrm{TV}}
\newcommand{\Bern}{\mathrm{Bern}}
\newcommand{\PMFEquals}{\textsc{PMFEquals}}

\newcommand{\SubsetProd}{\textsc{SubsetProd}}

\newcommand{\NP}{\mathsf{NP}}

\newcommand{\SZK}{\mathsf{SZK}}
\newcommand{\NISZK}{\mathsf{NISZK}}

\newcommand{\R}{\mathbb{R}}

\newcommand{\KNAPSACK}{\textsc{Knapsack}}
\newcommand{\ProdDtv}{\textsc{DtvProduct}}
\newcommand{\ProdDtvUnif}{\textsc{DtvProductUnif}}
\newcommand{\BayesDtv}{\textsc{DtvBayesNet}}

\newcommand{\tv}{d_{\mathrm{TV}}}
\newcommand{\eps}{\varepsilon}

\newcommand{\dtv}{\tv}
\newcommand{\ignore}[1]{}

\newcommand{\U}{\mathbb{U}}

\renewcommand{\P}{\mathsf{P}}

\allowdisplaybreaks

\title{\bf On Approximating Total Variation Distance}

\author{%
\bf Arnab Bhattacharyya
\\
School of Computing \\
National University of Singapore \\
\texttt{arnabb@nus.edu.sg}
\and
\bf Sutanu Gayen \\
CSE Department \\
Indian Institute of Technology Kanpur \\
\texttt{sutanu@cse.iitk.ac.in}
\and
\bf Kuldeep S. Meel \\
School of Computing \\
National University of Singapore \\
\texttt{meel@comp.nus.edu.sg}
\and
\bf Dimitrios Myrisiotis \\
School of Computing \\
National University of Singapore \\
\texttt{dimitris@nus.edu.sg}
\and
\bf A. Pavan \\
Department of Computer Science \\
Iowa State University \\
\texttt{pavan@iastate.edu}
\and
\bf N.~V.~Vinodchandran \\
School of Computing \\
University of Nebraska-Lincoln \\
\texttt{vinod@cse.unl.edu}
}

\begin{document}

\maketitle

\begin{abstract}
Total variation distance (TV distance) is a fundamental notion of distance between probability distributions.
In this work, we introduce and study the problem of computing the TV distance of two product distributions over the domain $\{0,1\}^n$.
In particular, we establish the following results.
\begin{enumerate}
\item
The problem of exactly computing the TV distance of two product distributions is $\#\P$-complete.
This is in stark contrast with other distance measures such as KL, Chi-square, and Hellinger which tensorize over the marginals leading to efficient algorithms.
\item
There is a fully polynomial-time deterministic approximation scheme (FPTAS) for computing the TV distance of two product distributions $P$ and $Q$ where $Q$ is the uniform distribution.
This result is extended to the case where $Q$ has a constant number of distinct marginals.
In contrast, we show that when $P$ and $Q$ are Bayes net distributions, the relative approximation of their TV distance is $\NP$-hard.
\end{enumerate}
\end{abstract}

\section{Introduction}

\label{sec:introduction}

An overarching theme in modern machine learning is the use of probability distributions to describe data.
Datasets are often modeled by high-dimensional distributions with additional structures reflecting correlations among the features.
In this context, a basic problem is {\em distance computation}: Given two distributions $P$ and $Q$, compute $\rho(P,Q)$ for a distance measure $\rho$.
For example, $P$ and $Q$ could be the outputs of two unsupervised learning algorithms, and one could ask how much they differ.
As another example, a key component of generative adversarial networks~\cite{goodfellow2014generative,arjovsky2017wasserstein} is the discriminant which approximates the distance between the model and the true distributions.

Given two distributions $P$ and $Q$ over a finite domain $\mathcal{D}$, their {\em total variation (TV) distance} or {\em statistical difference} $\dtv(P,Q)$ is defined as
\[
\dtv(P,Q)= \max_{S \subseteq \mathcal{D}} \cprn{P(S)-Q(S)}
= \frac12 \sum_{x \in \mathcal{D}} |P(x) - Q(x)|
\]
which is also equal to $\sum_{x \in \mathcal{D}} \max \cprn{0, P(x) - Q(x)}$.
The total variation distance satisfies certain fundamental properties.
First, it has a physical interpretation:
The TV distance between two distributions is the maximum bias of any event with respect to the two distributions.
Second, it satisfies many mathematically desirable properties:
It is bounded, it is a metric, and it is invariant with respect to bijections.
Because of these reasons, the total variation distance is one of the main distance measures employed in a wide range of areas including probability and statistics, machine learning, information theory, and pseudorandomness.

In this work, we study the total variation distance from a {\em computational} perspective.
Given two distributions $P$ and $Q$ over a finite domain $\mathcal{D}$, how hard is it to compute $\dtv(P,Q)$?
If $P$ and $Q$ are explicitly specified by the probabilities of all of the points of the (discrete) domain $\mathcal{D}$, summing up the absolute values of the differences in probabilities at all points leads to a simple linear time algorithm.
However, in many applications, the distributions of interest are of a high dimension with succinct representations.
In these scenarios, since the size of the domain $\mathcal{D}$ is very large, an $O(|\mathcal{D}|)$ algorithm is highly impractical.
Therefore, a fundamental computational question is:
\begin{quote}
Can we design efficient algorithms (with running time polynomial in the size of the representation) for computing the TV distance between two high-dimensional distributions with succinct representations?
\end{quote}
The simplest model for a high-dimensional distribution is the {\em product distribution}, which is a product of independent Bernoulli trials.
More precisely, a product distribution $P$ over $\mathcal{D} = \{0,1\}^n$ is succinctly described by $n$ parameters $p_1, \dots, p_n$ where each $p_i \in [0,1]$ is independently the probability that the $i$-th coordinate equals $1$.
Product distributions serve as a great testing ground for various intuitions regarding computational statistics, due to their ubiquity and simplicity.
Despite their simplicity, surprisingly little is known about the complexity of computing the TV distance between product distributions.
A very recent result shows the existence of a fully polynomial-time {\em randomized approximation} scheme (FPRAS) to relatively approximate the TV distance between two product distributions~\cite{fgjw22}.
However, this result does not shed light on the complexity of the exact computation of TV distance as well as the existence of {\em deterministic} approximation schemes (FPTAS).
Understanding the computational landscape of the total variation distance of product distributions is an important question. The present work makes significant progress towards this research goal. 

\subsection{Our Contributions}

Our contributions are the following:
\begin{enumerate}
\item
\label{item:hardness}
{\bf We show that the exact computation of the total variation distance between two product distributions $P$ and $Q$ is $\#\P$-complete (Theorem~\ref{thm:hardness-Bern-products-intro}).}
This hardness result holds even when the distribution $Q$ has at most $3$ distinct one-dimensional marginals.
Hence it is unlikely that there is an efficient algorithm for this computational problem, as an efficient algorithm for this problem would lead to efficient algorithms for many hard counting problems, including that of computing the number of satisfying assignments of a Boolean formula and all of the problems in the Polynomial-time Hierarchy~\cite{Sto76,Tod91}.

This is a surprising result, given that for many other distance measures such as Hellinger, Chi-square, and KL, there are efficient algorithms for computing the distance between two product distributions.
This is so, as these distances {\em tensorize} over their marginals (folklore; see also \cite{DBLP:conf/alt/0001G0V21}), in the sense that they are easily expressible in terms of their one-dimensional marginals.
\item
{\bf We design a fully polynomial-time {\em deterministic} approximation scheme (FPTAS) that computes a relative approximation of the TV distance between two product distributions $P$ and $Q$ where $Q$ is the uniform distribution (\Cref{thm:pqhalf}).}
Building on the techniques developed, we design an FPTAS for the case when $Q$ has a constant number of distinct one-dimensional marginals (\Cref{prop:Bern-a}; \Cref{prop:constant-q_i's}).
This, combined with the earlier-mentioned hardness result, completely characterizes the complexity of TV distance computation for product distributions when one of the distributions has a constant number of one-dimensional marginals.
\item
We investigate the complexity of the problem when the distributions $P$ and $Q$ are slightly more general than product distributions.
In particular, {\bf we show that it is $\NP$-hard to relatively approximate the TV distance between two sparse {\em Bayesian networks}~\cite{Pea89} (see \Cref{thm:dTV-zero-or-not-NP-complete-intro})}.
\end{enumerate}
In summary, our study showcases the rich complexity landscape of the problem of total variation distance computation, even for simple distributions.

\subsection{Organization}

The rest of the paper is organized as follows:
We present preliminaries in \Cref{sec:preliminaries} and discuss related work in \Cref{sec:relatedwork}.
We then present, in \Cref{sec:hardness}, the $\#\P$-hardness of the TV distance computation between product distributions.
In \Cref{sec:deterministic-approximation} we present our polynomial-time deterministic approximation schemes for the estimation of TV distance between some special cases of product distributions.
We conclude in \Cref{sec:conclusion}.
Finally, in \Cref{sec:appendix}, we present the deferred proof of \Cref{thm:GKM}.

\section{Preliminaries}

\label{sec:preliminaries}

We use $\sqbra{n}$ to denote the \emph{ordered} set $\brkts{1,\dots,n}$.
We will use $\log$ to denote $\log_2$ and $\U$ to denote the uniform distribution over the sample space.
Throughout the paper, we shall assume that all probabilities are represented as rational numbers of the form $a/b$.

A Bernoulli distribution with parameter $p$ is denoted by $\Bern\cprn{p}$.
A \emph{product distribution} is a product of independent Bernoulli distributions.
A product distribution $P$ over $\{0,1\}^n$ can be described by $n$ parameters $p_1,\dots,p_n$ where each $p_i \in [0,1]$ is the probability that the $i$-th coordinate equals $1$ (such a $P$ is usually denoted by $\bigotimes_{i=1}^n\Bern(p_i)$).
For any $x\in \{0,1\}^n$, the probability of $x$ with respect to the distribution $P$ is given by $P\cprn{x}=\prod_{i\in S}p_i\prod_{i\in\sqbra{n}\setminus S}\prn{1-p_i}\in\sqbra{0,1}$, where $S \subseteq \sqbra{n}$ is such that $i\in S$ if and only if the $i$-th coordinate of $x$ is $1$, independently.

$\ProdDtv$ is the following computational problem:
Given two product distributions $P$ and $Q$ over the sample space $\{0,1\}^n$, compute $\dtv(P,Q)$.
When the distribution $Q$ is the uniform distribution over $\{0,1\}^n$, we denote the above problem by $\ProdDtvUnif$.

A \emph{Bayes net} is specified by a directed acyclic graph (DAG) and a sequence of conditional probability tables (CPTs), one for each of its nodes (and for each setting of the parents of each node).
In this way, one may define a probability distribution over the nodes of a Bayes net.
We will also  consider the problem of computing $\dtv(P,Q)$ where $P$ and $Q$ are Bayes net distributions, which we denote by $\BayesDtv$.

A function $f$ from $\{0,1\}^*$ to non-negative integers is in the class $\#\P$ if there is a polynomial time non-deterministic Turing machine $M$ so that for any $x$, $f(x)$ is equal to the number of accepting paths of $M(x)$.
Our hardness result will make use of the known $\#\P$-complete problem $\#\SubsetProd$ which is a counting version of the $\NP$-complete problem $\SubsetProd$ (see \cite{GJ79}; the proof is attributed to Yao).
$\#\SubsetProd$ is the following problem:
Given integers $a_1,\dots,a_n$, and a target number $T$, compute the number of sets $S\subseteq\sqbra{n}$ such that $\prod_{i\in S}a_i=T$.

We also require a counting version of the $\KNAPSACK$ problem, $\#\KNAPSACK$ which is defined as follows:
Given weights $a_1,\dots,a_n$ and capacity $b$, compute the number of sets $S\subseteq\sqbra{n}$ such that $\sum_{i\in S}a_i\leq b$. It is known that $\#\KNAPSACK$ is $\#\P$-complete.

The following notion of an approximation algorithm is important in this work.

\begin{definition}
A function $f: \bool^* \to \R$ admits a \emph{fully polynomial-time approximation scheme (FPTAS)} if there is a {\em deterministic} algorithm $\mathcal{A}$ such that for every input $x$ (of length $n$) and $\epsilon >0$, $\mathcal{A}$ on inputs $x$ and $\eps$ outputs a $\prn{1+\eps}$-relative approximation of $f(x)$, i.e., a value $v$ that lies in the interval $[f(x)/(1+\eps), (1+\eps)f(x)]$.
The running time of $\mathcal{A}$ is polynomial in $n$, and $1/\eps$.
\end{definition}

We require the following result from \cite{GKM10}.

\begin{lemma}[\cite{GKM10}]
\label{lem:GKM-1} 
There is an FPTAS for $\#\KNAPSACK$.
\end{lemma}

In our work, we shall also use the following adaptation of the framework that was introduced by \cite{GKM10}.
We fix some terminology first.
For a set $S\subseteq\sqbra{n}$ its \emph{Hamming weight (or cardinality)} $\abs{S}$ is the number of $1$'s in its characteristic vector in $\bool^n$.
Given a vector $v$ in $\bool^n$ and a set $S=\brkts{i_1,\dots,i_k}\subseteq\sqbra{n}$, the \emph{projection of $v$ at $S$} is the string $v_{i_1}\cdots v_{i_k}$.

\begin{lemma}[Following~\cite{GKM10}; proof in \Cref{sec:appendix}]
\label{thm:GKM}
There is a deterministic algorithm that, given a $\#\KNAPSACK$ instance $(a_1,\dots,a_n,b)$ of total weight $W = \sum_i a_i + b$, $\delta > 0$, a $k$-size partition $S_1,\dots,S_k$ of $\sqbra{n}$ for some constant $k$, and $r_1,\dots,r_k\in [n]$ such that $r_i \leq |S_i|$, outputs a $(1+\delta)$-relative approximation of the number of $\KNAPSACK$ solutions such that their projections at sets $S_1,\dots,S_k$ have Hamming weights $r_1,\dots,r_k$, respectively.
The running time of this algorithm is polynomial in $n,\log W$, and $1/\delta$.
\end{lemma}

\section{Related Work}

\label{sec:relatedwork}

Most of the earlier works on computing the TV distance of succinctly represented high dimensional distributions are about the complexity and feasibility of {\em additive} approximations. Sahai and Vadhan~\cite{SV03} established in a seminal work that additively approximating the TV distance between two distributions that are samplable by Boolean circuits is hard for the complexity class $\SZK$ (Statistical Zero Knowledge).
The complexity class $\SZK$ is fundamental to cryptography and is believed to be computationally hard.
Subsequent works captured variations of this theme~\cite{GoldreichSV99,Malka15,DPV20}:
For example,~\cite{GoldreichSV99} showed that the problem of deciding whether a distribution samplable by a Boolean circuit is close or far from the uniform distribution is complete for the complexity class $\NISZK$ (Non-Interactive Statistical Zero Knowledge).
Another line of work focuses on finding the complexity of computing the TV distance between two hidden Markov models culminating in the results that it is undecidable whether the TV distance is greater than a threshold or not, and that it is $\#\P$-hard to additively approximate it~\cite{CortesMR07,LyngsoP02,hmm18}.

Complementing the above hardness results,~\cite{0001GMV20} designed efficient algorithms to additively approximate the TV distance of distributions that are efficiently samplable and also efficiently computable (meaning that their probability mass function is efficiently computable).
In particular, they designed efficient algorithms for {additively} approximating the TV distance of structured high dimensional distributions such as Bayesian networks, Ising models, and multivariate Gaussians.
In a similar vein, \cite{PM21} studied a related \emph{property testing} variant of TV distance, for distributions encoded by circuits.

Relative approximation of TV distance has received less attention compared to additive approximation.
Very recently, \cite{fgjw22} designed an FPRAS for relatively approximating the TV distance between two product distributions. The current work, in addition to showing that the exact computation of the TV distance between two product distributions is $\#\P$-complete, also presents deterministic approximation algorithms for a certain class of product distributions.

The work of \cite{fgjw22} relies on {\em coupling techniques} from probability theory (which appear inherently randomized), whereas we design deterministic algorithms via a reduction to $\#\KNAPSACK$ for which deterministic approximation schemes exist (see \Cref{sec:deterministic-approximation}):
\cite{DFK+93} gave a subexponential-time approximation algorithm for $\#\KNAPSACK$.
Later, \cite{MS04} designed an FPRAS for it.
Subsequently, \cite{Dye03} presented an FPRAS for $\#\KNAPSACK$ using simpler techniques.
Later independent works of Stefankovic, Vempala, and Vigoda~\cite{SVV12} and Gopalan, Klivans, and Meka~\cite{GKM10} gave FPTAS for $\#\KNAPSACK$.
Our work relies on the algorithms presented in~\cite{GKM10}.

Finally, a work that highlights some interesting aspects of product distributions is~\cite{SRG17}, whereby they show that computing $r$-th order statistics for product distributions is $\NP$-hard.

\section{The Hardness of Computing TV Distance}

\label{sec:hardness}

In this section, we establish hardness results.
We first show that $\ProdDtv$ is $\#\P$-complete.
Then we show that it is $\NP$-hard to approximate the TV-distance between distributions that are slightly more general than product distributions.
More specifically, we show that it is $\NP$-hard to design an approximation algorithm for $\BayesDtv$, even when the underlying Bayes nets are of in-degree two.

\subsection{\texorpdfstring{$\#\P$-Completeness of $\ProdDtv$}{}}

We establish that following result.

\begin{theorem}
\label{thm:hardness-Bern-products-intro}
$\ProdDtv$ is $\#\P$-complete.
This holds even when one of the distributions has at most $3$ distinct one-dimensional marginals.
\end{theorem}

\paragraph{Proof overview:}

We show the hardness in two steps.
In the first step, we introduce a problem called $\#\PMFEquals$ and show that it is $\#\P$-hard by a reduction from $\#\SubsetProd$.

$\#\PMFEquals$ is the following problem:
Given a probability vector $\prn{p_1,\dots,p_n}$ where $p_i\in\sqbra{0,1}$ and a number $v$, compute the number of $x\in\bool^n$ such that $P\cprn{x}=v$, where $P$ is the product distribution described by $\prn{p_1,\dots,p_n}$.

In the second step, we reduce $\#\PMFEquals$ to the problem of computing the TV distance of two product distributions.
For this, given a product distribution $P$, we construct product distributions $\hat P,\hat Q,P',Q'$ such that $\#\PMFEquals$ is a polynomial-time computable function of $\dtv\cprn{P',Q'}$ and $\dtv\cprn{\hat P,\hat Q}$.
In particular, we establish that for the case where $v<2^{-n}$ it is the case that $\abs{\brkts{x\mid P\cprn{x}=v}}$ is equal to
\[
\frac{d_{\TV}\cprn{P',Q'}-d_{\TV}\cprn{\hat P,\hat Q}}{2\beta v}
\]
for an appropriately chosen $\beta$ (similarly for the case where $v\geq2^{-n}$).
Thus if there is an efficient algorithm for $\ProdDtv$, then that algorithm can be used to efficiently solve the $\#\P$-complete problem $\#\SubsetProd$.

\paragraph{Detailed proof:}

We begin with the $\#\P$-hardness of $\#\PMFEquals$.

\begin{lemma}
\label{lemma:PMFEquals-is-hard}
$\#\PMFEquals$ is $\#\P$-hard.
\end{lemma}

\begin{proof}
We will reduce $\#\SubsetProd$ to $\#\PMFEquals$.
The result will then follow from the fact that $\#\SubsetProd$ is $\#\P$-hard.
Let $a_1,\dots,a_n$, and $T$ be the numbers of an arbitrary $\#\SubsetProd$ instance, namely $I_S$.
We will create a $\#\PMFEquals$ instance $I_P$ that has the same number of solutions as $I_S$.

Let $p_i:=\frac{{a_i}}{1+{a_i}}$ for every $i$ and $v:=T\prod_{i\in\sqbra{n}}(1-p_i)$, and observe that $a_i=\frac{p_i}{1-p_i}$.
For any set $S\subseteq [n]$, we have the following equivalences:
\[
\prod_{i\in S}a_i=T
\Leftrightarrow\prod_{i\in S}\frac{p_i}{1-p_i}=\frac{v}{\prod_{i\in\sqbra{n}}(1-p_i)}
\Leftrightarrow\prod_{i\in S}p_i\prod_{i\notin S}\prn{1-p_i}=v
\Leftrightarrow P\cprn{x}=v,
\]
where $x$ is such that $x_i=1$ if and only if $i\in S$.
This completes the proof.
\end{proof}

We now turn to \Cref{thm:hardness-Bern-products-intro}.

\begin{proof}[Proof of \Cref{thm:hardness-Bern-products-intro}]
We separately prove membership in $\#\P$ and $\#\P$-hardness.

\paragraph{Membership in $\#\P$:}

Let $P$ and $Q$ be two product distributions, specified by parameters $p_1,\dots,p_n$ and $q_1,\dots,q_n$, respectively.
Without loss of generality we shall assume that these parameters are fractions (as we only have some finite precision available).
The goal is to design a nondeterministic machine ${\cal N}$ that takes $p_1,\dots,p_n$ and $q_1,\dots,q_n$ as inputs and is such that the number of its accepting paths (normalized by an appropriate quantity; see~\cite{CSW20}) equals $\dtv(P,Q)$.

Let $M$ be the product of the denominators of all parameters $p_1,\dots,p_n,q_1,\dots,q_n$ and their complements $1-p_1,\dots,1-p_n,1-q_1,\dots,1-q_n$.
The non-deterministic machine ${\cal N}$ first guesses an element $i\in\bool^n$ in the sample space of $P$ and $Q$, computes $|P(i)-Q(i)|$ by using the parameters $p_1,\dots,p_n,q_1,\dots,q_n$, then guesses an integer $0\leq z\leq M$, and finally accepts if and only if $1\leq z\leq M|P(i)-Q(i)|$.
(Note that $M|P(i)-Q(i)|=\abs{M\cdot P(i)-M\cdot Q(i)}$ is an integer.)

It follows that
\[
\dtv(P,Q)
=\frac{1}{2}\sum_{i\in\bool^n}\abs{P\cprn{i}-Q\cprn{i}}
=\frac{\text{number of accepting paths of }{\cal N}}{2M}
\]
since the number of accepting paths of ${\cal N}$ is
\[
\sum_{i\in\bool^n}\prn{M\abs{P\cprn{i}-Q\cprn{i}}}
=M\sum_{i\in\bool^n}\abs{P\cprn{i}-Q\cprn{i}}
=M\cdot2\dtv(P,Q).
\]

\paragraph{$\#\P$-hardness:}

For establishing hardness, we will reduce $\#\PMFEquals$ to $\ProdDtv$.
The theorem will then follow from \Cref{lemma:PMFEquals-is-hard}.

Let $p_1,\dots,p_n$ and $v$ be the numbers in an arbitrary instance of $\#\PMFEquals$ where each $p_i$ is represented as an $m$-bit binary fraction.
With this, $P(x)$ can be represented as an $nm$-bit binary fraction.
Thus without loss of generality, we can assume that $v$ is also an $nm$-bit fraction.
We distinguish between two cases depending on whether $v<2^{-n}$ or $v\geq2^{-n}$.

\subparagraph{Case A: $v<2^{-n}$.}

First, we construct two distributions $\hat P=\Bern\cprn{\hat p_1}\otimes\cdots\otimes \Bern\cprn{\hat p_{n+1}}$ and $\hat Q=\Bern\cprn{\hat q_1}\otimes\cdots\otimes\Bern\cprn{\hat q_{n+1}}$ over $\{0,1\}^{n+1}$ as follows:
$\hat p_i:=p_i$ for $i\in\sqbra{n}$ and $\hat p_{n+1}:=1$; $\hat q_i:=1/2$ for $i\in\sqbra{n}$ and $\hat q_{n+1}:=v2^n$.
We have that $\dtv\cprn{\hat P,\hat Q}$ is equal to
\begin{align}
\sum_{x\in\{0,1\}^{n+1}}\max\cprn{0,\hat P\cprn{x}-\hat Q\cprn{x}}
&=\sum_{x}\max\cprn{0,P\cprn{x}-\frac{1}{2^n}v2^n}\nonumber\\
&=\sum_{x\in\bool^n}\max\cprn{0,P\cprn{x}-v}\nonumber\\
&=\sum_{x:P\prn{x}>v}\prn{P\cprn{x}-v}.
\label{eq:hardness}
\end{align}
We now define two more distributions $P'$ and $Q'$ over $\{0,1\}^{n+2}$, by making use of the following claim.

\begin{claim}
\label{clm:beta}
There exists a $\beta\in\prn{0,1}$ such that the following hold for all $x$:
If $P\cprn{x}<v$, then $P\cprn{x}\prn{\frac{1}{2}+\beta}<v\prn{\frac{1}{2}-\beta}$;
if $P\cprn{x}>v$, then $P\cprn{x}\prn{\frac{1}{2}-\beta}>v\prn{\frac{1}{2}+\beta}$.
In particular, we can take $\beta$ to be equal to $\frac{1}{2^{3nm}}$.
\end{claim}

\begin{proof}
For \Cref{clm:beta} to hold, observe that we want $\beta$ to be at most $\frac{|v-P(x)|}{v+P(x)}$ for every $x$, so that $P(x) \neq v$.
Since both $v$ and $P(x)$ have $nm$-bit representations, both $|v-P(x))|$ and $v+P(x)$ have $nm$-bit representations. 
Thus $\frac{|v-P(x)|}{v+P(x)}$ can be represented as a $2nm$-bit fraction.
Since this fraction is not zero, and the smallest $2nm$-bit fraction is $\frac{1}{2^{2nm}}$, choosing $\beta := \frac{1}{2^{3nm}}$ suffices.
\end{proof}

We now define two new distributions $P'$ and $Q'$ as follows:
$p_i':=p_i$ for $i\in [n]$, $p_{n+1}':=1$, and $p_{n+2}':=\frac{1}{2}+\beta$; $q_i':=\frac{1}{2}$ for $i\in [n]$, $q_{n+1}':=v2^n$, and $q_{n+2}':=\frac{1}{2}-\beta$ where $\beta$ is as in \Cref{clm:beta}.

We establish the following claim.

\begin{claim}
\label{clm:v-small}
It is the case that
$\abs{\brkts{x\mid P\cprn{x}=v}}$
equals
\[
\frac{d_{\TV}\cprn{P',Q'}-d_{\TV}\cprn{\hat P,\hat Q}}{2\beta v}.
\]
\end{claim}

\begin{proof}
We have that $d_{\TV}\cprn{P',Q'}$ is equal to
\begin{align*}
\sum_{x\in \{0, 1\}^{n+2}}\max\cprn{0,P'\cprn{x}-Q'\cprn{x}}
&=\sum_{x\in\{0,1\}^n}\max\cprn{0,P\cprn{x}\prn{\frac{1}{2}+\beta}-\frac{1}{2^n}v2^n\prn{\frac{1}{2}-\beta}}\\
&\qquad+\sum_{x\in\{0,1\}^n}\max\cprn{0,P\cprn{x}\prn{\frac{1}{2}-\beta}-\frac{1}{2^n}v2^n\prn{\frac{1}{2}+\beta}}\\
&=\sum_{x}\max\cprn{0,P\cprn{x}\prn{\frac{1}{2}+\beta}-v\prn{\frac{1}{2}-\beta}}\\
&\qquad+\sum_{x}\max\cprn{0,P\cprn{x}\prn{\frac{1}{2}-\beta}-v\prn{\frac{1}{2}+\beta}}\\
&=\sum_{x:P\prn{x}\geq v} P\cprn{x}\prn{\frac{1}{2}+\beta}-v\prn{\frac{1}{2}-\beta}\\
&\qquad+\sum_{x:P\prn{x}>v}P\cprn{x}\prn{\frac{1}{2}-\beta}-v\prn{\frac{1}{2}+\beta}\\
&=\sum_{x:P\prn{x}= v}P\cprn{x}\prn{\frac{1}{2}+\beta}-v\prn{\frac{1}{2}-\beta}\\
&\qquad+\sum_{x:P\prn{x}> v} P\cprn{x}\prn{\frac{1}{2}+\beta}-v\prn{\frac{1}{2}-\beta}\\
&\qquad+\sum_{x:P\prn{x}>v}P\cprn{x}\prn{\frac{1}{2}-\beta}-v\prn{\frac{1}{2}+\beta}\\
&=2\beta v\abs{\brkts{x\mid P\cprn{x}=v}}
+\sum_{x:P\prn{x}>v}\prn{P\cprn{x}-v}.
\end{align*}
The result now follows from \Cref{eq:hardness}.
The first equality follows from the definitions of $P'$ and $Q'$ (since $p_{n+1}'= 1$).
Note that when for every $x$ with $P(x) < v$, by Claim~\ref{clm:beta}, $P(x)(\frac{1}{2}+\beta) < v(\frac{1}{2}-\beta)$ and if $P(x) \geq v$, then $P(x)(\frac{1}{2}+\beta) \geq  v(\frac{1}{2}-\beta)$.
Also when $P(x) \leq  v$, $P(x)(\frac{1}{2}+\beta)$ is at most $v(\frac{1}{2}-\beta)$ and when $P(x) > v$, by Claim~\ref{clm:beta}, we have $P(x)(\frac{1}{2}-\beta) > v(\frac{1}{2}+\beta)$.
These imply the third equality.
The rest of the equalities holds by algebraic manipulations.
\end{proof}

For that matter $|\{x\mid P(x) = v\}|$ can be computed by 
computing $\dtv(P',Q')$ and $\dtv(\hat{P},\hat{Q})$. 
Thus the proof in this case follows by \Cref{lemma:PMFEquals-is-hard}.

\subparagraph{Case B: $v \geq 2^{-n}$.}

First, let us define distributions $\hat P=\Bern\cprn{\hat p_1}\otimes\cdots\otimes\Bern\cprn{\hat p_n}$ and $\hat Q=\Bern\cprn{\hat q_1}\otimes\cdots\otimes\Bern\cprn{\hat q_n}$ as follows: $\hat p_i:=p_i$ for $i\in\sqbra{n}$, $\hat p_{n+1}:=\frac{1}{v2^n}$; $\hat q_i:=\frac{1}{2}$ for $i\in\sqbra{n}$, and $\hat q_{n+1}:=1$.

We now have that $d_{\TV}\cprn{\hat P,\hat Q}$ is equal to $\frac{1}{2}\sum_{x}\abs{\hat P\cprn{x}-\hat Q\cprn{x}}$ or
\begin{align*}
\sum_{x}\max\cprn{0,\hat P\cprn{x}-\hat Q\cprn{x}}
&=\sum_{x}\max\cprn{0,P\cprn{x}\frac{1}{v2^n}-\frac{1}{2^n}}
+\sum_{x}\max\cprn{0,P\cprn{x}\prn{1-\frac{1}{v2^n}}}\\
&=\sum_{x}\max\cprn{0,P\cprn{x}\frac{1}{v2^n}-\frac{1}{2^n}}
+1-\frac{1}{v2^n}.
\end{align*}
As earlier, we define two more distributions $P'$ and $Q'$, by making use of \Cref{clm:beta}.
The new distributions $P'$ and $Q'$ are such that $p_i':=p_i$ for $i\in\sqbra{n}$, $p_{n+1}':=\frac{1}{v2^n}$, and $p_{n+2}':=\frac{1}{2}+\beta$; $q_i':=1/2$ for $i\in\sqbra{n}$, $q_{n+1}':=1$, and $q_{n+2}':=\frac{1}{2}-\beta$.

We establish the following claim.

\begin{claim}
\label{clm:v-big}
We have that $\abs{\brkts{x\mid P\cprn{x}=v}}$ is equal to
\[
\frac{2^{n-1}}{\beta}\prn{d_{\TV}\cprn{P',Q'}-d_{\TV}\cprn{\hat P,\hat Q}}.
\]
\end{claim}

\begin{proof}
By our previous discussion we know that $d_{\TV}\cprn{P',Q'}$ is equal to
\begin{align*}
\sum_{x}\max\cprn{0,P'\cprn{x}-Q'\cprn{x}}
&=\sum_{x}\max\cprn{0,P\cprn{x}\frac{1}{v2^n}\prn{\frac{1}{2}+\beta}-\frac{1}{2^n}\prn{\frac{1}{2}-\beta}}\\
&\qquad+\sum_{x}\max\cprn{0,P\cprn{x}\frac{1}{v2^n}\prn{\frac{1}{2}-\beta}-\frac{1}{2^n}\prn{\frac{1}{2}+\beta}}\\
&\qquad+\sum_{x}\max\cprn{0,P\cprn{x}\prn{1-\frac{1}{v2^n}}\prn{\frac{1}{2}+\beta}}\\
&\qquad+\sum_{x}\max\cprn{0,P\cprn{x}\prn{1-\frac{1}{v2^n}}\prn{\frac{1}{2}-\beta}}\\
&=\sum_{x}\max\cprn{0,P\cprn{x}\frac{1}{v2^n}\prn{\frac{1}{2}+\beta}-\frac{1}{2^n}\prn{\frac{1}{2}-\beta}}\\
&\qquad+\sum_{x}\max\cprn{0,P\cprn{x}\frac{1}{v2^n}\prn{\frac{1}{2}-\beta}-\frac{1}{2^n}\prn{\frac{1}{2}+\beta}}\\
&\qquad+\sum_{x}\max\cprn{0,P\cprn{x}\prn{1-\frac{1}{v2^n}}}\\
&=\sum_{x:P\prn{x}\geq v}\max\cprn{0,P\cprn{x}\frac{1}{v2^n}\prn{\frac{1}{2}+\beta}-\frac{1}{2^n}\prn{\frac{1}{2}-\beta}}\\
&\qquad+\sum_{x:P\prn{x}>v}\max\cprn{0,P\cprn{x}\frac{1}{v2^n}\prn{\frac{1}{2}-\beta}-\frac{1}{2^n}\prn{\frac{1}{2}+\beta}}\\
&\qquad+1-\frac{1}{v2^n},
\intertext{by \Cref{clm:beta}, or}
\sum_{x}\max\cprn{0,P'\cprn{x}-Q'\cprn{x}}
&=\sum_{x:P\prn{x}=v}\max\cprn{0,P\cprn{x}\frac{1}{v2^n}\prn{\frac{1}{2}+\beta}-\frac{1}{2^n}\prn{\frac{1}{2}-\beta}}\\
&\qquad+\sum_{x:P\prn{x}>v}\max\cprn{0,P\cprn{x}\frac{1}{v2^n}\prn{\frac{1}{2}+\beta}-\frac{1}{2^n}\prn{\frac{1}{2}-\beta}}\\
&\qquad+\sum_{x:P\prn{x}>v}\max\cprn{0,P\cprn{x}\frac{1}{v2^n}\prn{\frac{1}{2}-\beta}-\frac{1}{2^n}\prn{\frac{1}{2}+\beta}}\\
&\qquad+1-\frac{1}{v2^n}\\
&=2\beta\frac{1}{2^n}\abs{\brkts{x\mid P\cprn{x}=v}}
+\sum_{x}\max\cprn{0,P\cprn{x}\frac{1}{v2^n}-\frac{1}{2^n}}
+1-\frac{1}{v2^n}\\
&=\frac{\beta}{2^{n-1}}\abs{\brkts{x\mid P\cprn{x}=v}}+d_{\TV}\cprn{\hat P,\hat Q}
\end{align*}
by \Cref{clm:v-big}.
That is, $\abs{\brkts{x\mid P\cprn{x}=v}}$ is equal to
\[
\frac{2^{n-1}\prn{d_{\TV}\cprn{P',Q'}-d_{\TV}\cprn{\hat P,\hat Q}}}{\beta}.
\qedhere
\]
\end{proof}

Thus also in this case the proof follows by \Cref{lemma:PMFEquals-is-hard}.
Finally, note that in either case the distribution $\hat{Q}$ has $2$ distinct one-dimensional marginals and $Q'$ has $3$ distinct one-dimensional marginals.
\end{proof}

\subsection{Hardness of Approximating \texorpdfstring{$\BayesDtv$}{}}

In this section, we prove the following.

\begin{theorem}
\label{thm:dTV-zero-or-not-NP-complete-intro}
Given two probability distributions $P$ and $Q$ that are defined by Bayes nets of in-degree at least two, it is $\NP$-complete to decide whether $\dtv(P,Q)\neq 0$ or not.
Hence the problem of relatively approximating $\BayesDtv$ is $\NP$-hard.
\end{theorem}

\begin{proof}
The proof gives a reduction from the satisfiability problem for CNF formulas (which is $\NP$-hard~\cite{Coo71}) to deciding whether the total variation distance between two Bayes nets distributions is non-zero or not.
Let $F$ be a CNF formula viewed as a Boolean circuit.
Assume $F$ has $n$ input variables $x_1,\ldots, x_n$ and $m$ gates $\Gamma=\brkts{y_1,\ldots,y_m}$, where $\Gamma$ is topologically sorted with $y_m$ being the output gate.
We will define two Bayes net distributions on the same directed acyclic graph $G$ which, intuitively, is the graph of $F$.
(By a graph of a formula we mean the directed acyclic graph that captures the circuit structure of $F$, whereby the nodes are either AND, OR, NOT, or variable gates, and the edges correspond to wires connecting the gates.)

The vertex set of $G$ is split into two sets $\mathcal{X}$ and $\mathcal{Y}$, and a node $Z$.
The set ${\mathcal X}=\brkts{X_i}_{i=1}^n$ contains $n$ nodes with node $X_i$ corresponding to variable $x_i$ and the set ${\mathcal Y}=\brkts{Y_i}_{i=1}^m$ contains $m$ nodes with each node $Y_i$ corresponding to gate $y_i$.
So totally there are $n+m+1$ nodes.
There is directed edge from node $V_i$ to node $V_j$ if the gate/variable corresponding to $V_i$ is an input to $V_j$.

The distributions $P$ and $Q$ on $G$ are given by CPTs defined as follows.
Each $X_i$ is a uniformly random bit.
For each $Y_i$, its conditional probability table (CPT) is deterministic:
For each of the setting of the parents $Y_j, Y_k$ the variable $Y_i$ takes the value of the gate $y_i$ for that setting of its inputs $y_j,y_k$.
Finally, in the distribution $P$ the variable $Z$ is a random bit and in the distribution $Q$ the variable $Z$ is defined by the value of $Y_m$ OR-ed with a random bit.

Note that the formula $F$ computes a Boolean function on the input variables.
Let $f:\bool^n\to\bool$ be this function.
We extend $f$ to $\{0,1\}^{m}$ (i.e., $f:\bool^n\to\bool^m$) to also include the values of the intermediate gates.

With this notation for any binary string $XYZ$ of length $n+m+1$, both $P$ and $Q$ have a probability $0$ if $Y \neq f(X)$.
(In the derivation of TV distance that follows, we shall assume that $Y=f\cprn{X}$.)
Let $A:=\brkts{x\mid F\cprn{x}=1}$ and $R:=\brkts{x\mid F\cprn{x}=0}$.
Thus $2\dtv(P,Q)$ can be written as
\begin{align*}
\sum_{X,f(X),Z} |P - Q|
= \sum_{X\in A,Z} |P - Q| + \sum_{X\in R,Z} |P-Q|
\end{align*}
where we have abused the notation $P$ and $Q$ to denote the probabilities $P\cprn{X,f\cprn{X},Z}$ and $Q\cprn{X,f\cprn{X},Z}$, respectively.

We will now compute each sum separately.
First, we have that $\sum_{X \in A,Z} |P - Q|$ is equal to $\sum_{X\in A, Z=0} |P-Q| + \sum_{X\in A, Z=1} |P-Q|$ (taking cases for the value of $Z$) or
\[
\sum_{X \in A, Z=0} \abs{\frac{1}{2^{n+1}} - 0}
+ \sum_{X \in A, Z=1} \abs{\frac{1}{2^{n+1}} - \frac{1}{2^n}}
\]
which is equal to $\frac{|A|}{2^n}$; then, we have that the quantity $\sum_{X \in R,Z} |P - Q|$ is equal to $\sum_{X\in R, Z=0} |P-Q| + \sum_{X\in R, Z=1} |P-Q|$ (taking cases for the value of $Z$) or
\[
\sum_{X \in R, Z=0} \abs{\frac{1}{2^{n+1}} - \frac{1}{2^{n+1}}}
+ \sum_{X \in R, Z=1} \abs{\frac{1}{2^{n+1}} - \frac{1}{2^{n+1}}}
\]
which is equal to $0$.

Therefore $\dtv(P,Q) = {|A|}/{2^{n+1}}$.
The membership in $\NP$ follows because $\dtv(P,Q)\neq 0$ if and only if there is an $X$ so that $P(X) \neq Q(X)$; this can be checked in polynomial time for Bayes distributions over finite alphabets.
The $\NP$-hardness follows because the arbitarry CNF formula $F$ is satisfiable if and only if $|A|\neq 0$ if and only if $\dtv(P,Q)\neq 0$.

The $\NP$-hardness of relative approximation of $\BayesDtv$ follows as a relative approximation $\dtv(P,Q)$ is non-zero if and only if $\dtv(P,Q)\neq 0$.
\end{proof}

\section{Deterministic Approximation Schemes}

\label{sec:deterministic-approximation}

It is an open problem to design an FPTAS for \ProdDtv. 
In this section, we report progress on this by designing {\em deterministic} approximation algorithms for a few interesting subcases.
In particular, we first provide an FPTAS for computing the total variation distance between an arbitrary product distribution $P$ and the uniform distribution $\U$, and then extend to the case where $Q$ has $O(1)$ distinct $q_i's$.

\subsection{Algorithm for \ProdDtvUnif}

\label{sec:FPTAS-Uniform}

We establish the following theorem.

\begin{theorem}
\label{thm:pqhalf}
There is an FPTAS for $\dtv(P,\U)$ where $P=\Bern(p_1)\otimes\cdots\otimes\Bern(p_n)$.
\end{theorem}

\paragraph{Proof overview:}

The idea is to reduce an instance of \ProdDtvUnif\ to several instances of $\#\KNAPSACK$.
Since the latter problem has an FPTAS (\Cref{lem:GKM-1}), the theorem follows.

For every subset $S \subseteq [n]$, we assign a non-negative weight $Y_S$ and show that a ``normalized'' $\dtv(P, \U)$ is equal to $\sum_S Y_S$.
We express the problem of computing this summation as multiple \#\KNAPSACK\ instances.

For this, we first show that each non-zero $Y_S$ lies in the range $[1, V)$ (for an appropriate $V$ that depends on the granularity of our precision).
We divide the interval $[1, V)$ into subintervals of the form $\sqprn{(1+\eps)^{i-1},(1+\eps)^i}$ for various (yet polynomially many) values of $i$.
Let $k_i$ be the number of sets $S$ for which $Y_S$ lies in the $i$-th interval.
Then $\sum_{i\in[\poly{n}]} k_i (1+\varepsilon)^i$ yields an $(1+\eps)$ approximation of the normalized $\dtv\cprn{P,\U}$.
However, computing each $k_i$ exactly is also $\#\P$-hard.
Thus we seek an approximation of each $k_i$.

We use a re-organization trick of summations and additional techniques to express this as several $\#\KNAPSACK$ instances.
Setting the approximation parameter for $\#\KNAPSACK$ to $\eps$, this leads to a $(1+\eps^2)$-approximation algorithm.
By setting $\eps:=\delta/2$, we get an $(1+\delta)$-approximation algorithm.

\paragraph{Detailed proof:}

We now give detailed technical proof.
First we assume, without loss of generality, that no $p_i$ is equal to $1/2$ since otherwise we can ignore these coordinates $i$.
Moreover, again without loss of generality, we assume that $p_i>1/2$ for all $i$, since otherwise we can flip $0$ and $1$ in the $i$-th coordinate of both $P$ and $\U$.

Let $M$ be the set of indices $i\in\sqbra{n}$ such that $p_i=1$. Let $A:=\prod_{i\notin M}(1-p_i)$ and $W:=\frac{1}{2^n}\prod_{i\notin M} \frac{1}{1-p_i}$ be constants, and $W_{S}:=\prod_{i\in S\setminus M}\frac{p_i}{1-p_i}$, $W_{\emptyset}:=1$.

\begin{claim}
It is the case that $\dtv(P, \U)$ is equal to $A\cdot\sum_{S\subseteq [n]:M\subseteq S}\max\left(0,W_S-W\right)$.
\end{claim}

\begin{proof}
We have that $d_{\TV}\cprn{P,\U}$ equals $\sum_{x \in \{0,1\}^n} \max(0,P(x)-\U(x))$ or
\begin{align*}
&\sum_{S\subseteq [n]}\max\cprn{0,\prod_{i\in S}p_i\prod_{i\notin S}(1-p_i) - \frac{1}{2^n}}\\
&\qquad=\sum_{S\subseteq[n]:M\subseteq S}\max\cprn{0,\prod_{i\in S}p_i\prod_{i\notin S}(1-p_i) - \frac{1}{2^n}}\\
&\qquad=\prod_{i\notin M}(1-p_i)
\cdot\sum_{S\subseteq [n]:M\subseteq S}\max\!\left(0,\prod_{i\in S\setminus M}\frac{p_i}{1-p_i}-\frac{1}{2^n}\prod_{i\notin M} \frac{1}{1-p_i}\right)\\
&\qquad=A\sum_{S\subseteq [n]:M\subseteq S}\max\left(0,W_S-W\right).
\end{align*}
The second equality holds by the definition of $M$.
The third equality holds as
$\prod_{i\in S}{p_i}
=\prod_{i\in S\setminus M}{p_i}\prod_{i\in S\,\cap\, M}{p_i}
=\prod_{i\in S\setminus M}{p_i}$.
\end{proof}

For notational simplicity, we assume that each $p_i$ is represented using $\ell:=\cpoly{n}$ bits.
Thus each non-zero term $\max(0,P(x)-\U(x))$ of $\dtv\cprn{P,\U}$ contributes at least
\(
m_0:=2^{-\ell}=2^{-\poly{n}}
\)
to $\dtv\cprn{P,\U}$.
Hence for any $S$ for which $\max\left(0,W_S-W\right)>0$, its value is at least $m_{\min}:={m_0}/{A}\geq2^{-\poly{n}}$.
Moreover, $\max\left(0,W_S-W\right)$ is at most $m_{\max}$, defined as
\[
W_S
=\prod_{i\in S\setminus M}\frac{p_i}{1-p_i}
\leq\prod_{i\in S\setminus M}\frac{1-2^{-\poly{n}}}{2^{-\poly{n}}}
\leq {2^{\poly{n}}}
\]
by the facts that $p_i\leq1-2^{-\poly{n}}$ for $i\notin M$ (since we use some finite precision of $\cpoly{n}$ bits), $\prn{1-x}/x$ is non-decreasing in $x$, and $\prn{1-x}/x\leq1/x$ for all $x$.
Therefore $m_{\max}\leq2^{\poly{n}}$.

Consider now $Y_S:=\max\left(0,W_S-W\right)/{m_{\min}}$ which lies in $\sqprn{1,V}$ for some $V\leq m_{\max}/m_{\min}\leq2^{\poly{n}}$, and let
\[
\sqprn{1,V}=\bigcup_{i=0}^{u-1}\sqprn{(1+\eps)^{i},(1+\eps)^{i+1}}
\]
be a set of subintervals for integers $0 \le i \le u-1=\lceil \log_{1+\varepsilon} V\rceil-1\leq\cpoly{n}-1\leq\cpoly{n}$ and some $0<\varepsilon<1$ that we will fix later (as a function of $\delta$).

Let the number of sets $S$ such that $Y_S$ is in $\sqprn{1,(1+\eps)^{i}}$ be $n_i$.
Let the average contribution in the range $\sqprn{(1+\eps)^{i-1},(1+\eps)^{i}}$ be $B_i$.
We have the following equation:
\begin{align}
\frac{\dtv(P,\U)}{A\cdot m_{\min}}
=n_1B_1+(n_2-n_1)B_2
+(n_3-n_2)B_3+\dots+(n_u-n_{u-1})B_u.
\label{eq:dtv-approx}
\end{align}
Since $(1+\eps)^{i-1} \le B_i < (1+\eps)^i$, the following estimate $d$ is a $(1+\eps)$-approximation of the RHS:
\begin{align}
d:=n_1(1+\eps)+(n_2-n_1)(1+\eps)^2+(n_3-n_2)(1+\eps)^3
+\dots+(n_u-n_{u-1})(1+\eps)^u.
\label{eq:d-def}
\end{align}
We use a reorganization trick similar to~\cite{CM19}; see \Cref{fig:reorganization}.

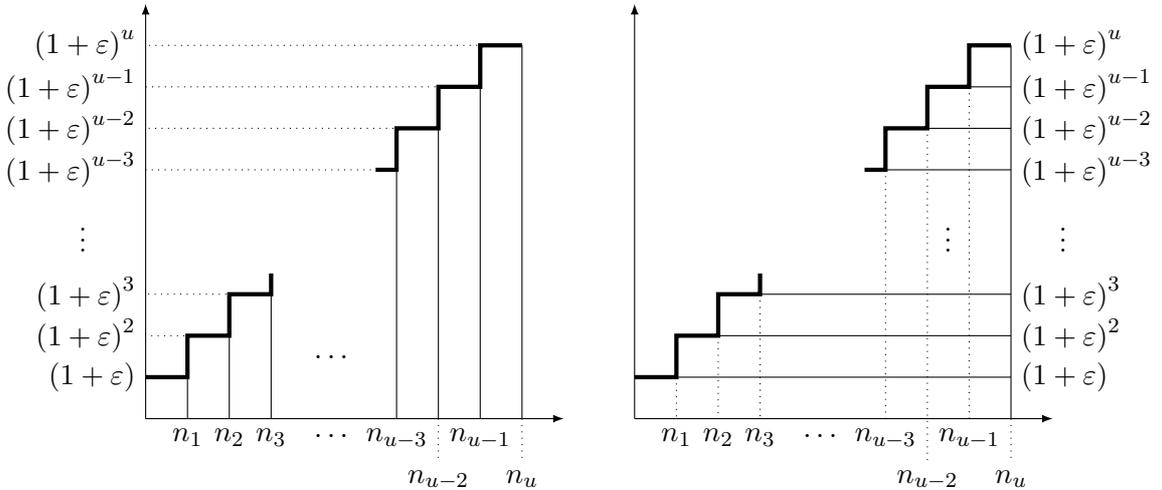
\begin{figure}[ht]
\centering
\begin{tikzpicture}[scale=0.55]
\draw[-latex] (0,0) -> (0,10);
\draw[-latex] (0,0) -> (10,0);
\draw[ultra thick] (0,1) -- (1,1) -- (1,2) -- (2,2) -- (2,3) -- (3,3) -- (3,3.5);
\draw[ultra thick] (5.5,6) -- (6,6) -- (6,7) -- (7,7) -- (7,8) -- (8,8) -- (8,9) -- (9,9);
\draw (9,9) -- (9,0);
\draw (1,1) -- (1,0);
\draw (2,2) -- (2,0);
\draw (3,3) -- (3,0);
\draw (6,6) -- (6,0);
\draw (7,7) -- (7,0);
\draw (8,8) -- (8,0);
\node at (4.5,1.5) {$\dots$};
\node at (4.5,-0.4) {$\dots$};
\node at (-1.5,4.5) {$\vdots$};
\node[below] at (1,0) {$n_1$};
\node[below] at (2,0) {$n_2$};
\node[below] at (3,0) {$n_3$};
\node[below] at (6,0) {$n_{u-3}$};
\node[below] at (7,-1) {$n_{u-2}$};
\node[below] at (8,0) {$n_{u-1}$};
\node[below] at (9,-1) {$n_u$};
\draw[dotted] (9,0) -- (9,-1);
\draw[dotted] (7,0) -- (7,-1);
\node[left] at (0,1) {$\prn{1+\varepsilon}$};
\node[left] at (0,2) {$\prn{1+\varepsilon}^2$};
\node[left] at (0,3) {$\prn{1+\varepsilon}^3$};
\node[left] at (0,6) {$\prn{1+\varepsilon}^{u-3}$};
\node[left] at (0,7) {$\prn{1+\varepsilon}^{u-2}$};
\node[left] at (0,8) {$\prn{1+\varepsilon}^{u-1}$};
\node[left] at (0,9) {$\prn{1+\varepsilon}^{u}$};
\draw[dotted] (9,9) -- (0,9);
\draw[dotted] (7,7) -- (0,7);
\draw[dotted] (8,8) -- (0,8);
\draw[dotted] (6,6) -- (0,6);
\draw[dotted] (2,2) -- (0,2);
\draw[dotted] (3,3) -- (0,3);
\end{tikzpicture}
\qquad
\begin{tikzpicture}[scale=0.55]
\draw[-latex] (0,0) -> (0,10);
\draw[-latex] (0,0) -> (10,0);
\draw[ultra thick] (0,1) -- (1,1) -- (1,2) -- (2,2) -- (2,3) -- (3,3) -- (3,3.5);
\draw[ultra thick] (5.5,6) -- (6,6) -- (6,7) -- (7,7) -- (7,8) -- (8,8) -- (8,9) -- (9,9);
\draw (9,9) -- (9,0);
\draw (1,1) -- (9,1);
\draw (2,2) -- (9,2);
\draw (3,3) -- (9,3);
\draw (6,6) -- (9,6);
\draw (7,7) -- (9,7);
\draw (8,8) -- (9,8);
\node at (7.5,4.5) {$\vdots$};
\node at (4.5,-0.4) {$\dots$};
\node at (10.25,4.5) {$\vdots$};
\node[below] at (1,0) {$n_1$};
\node[below] at (2,0) {$n_2$};
\node[below] at (3,0) {$n_3$};
\node[below] at (6,0) {$n_{u-3}$};
\node[below] at (7,-1) {$n_{u-2}$};
\node[below] at (8,0) {$n_{u-1}$};
\node[below] at (9,-1) {$n_u$};
\draw[dotted] (9,0) -- (9,-1);
\draw[dotted] (7,7) -- (7,-1);
\draw[dotted] (8,8) -- (8,0);
\draw[dotted] (6,6) -- (6,0);
\draw[dotted] (1,1) -- (1,0);
\draw[dotted] (2,2) -- (2,0);
\draw[dotted] (3,3) -- (3,0);
\node[right] at (9,1) {$\prn{1+\varepsilon}$};
\node[right] at (9,2) {$\prn{1+\varepsilon}^2$};
\node[right] at (9,3) {$\prn{1+\varepsilon}^3$};
\node[right] at (9,6) {$\prn{1+\varepsilon}^{u-3}$};
\node[right] at (9,7) {$\prn{1+\varepsilon}^{u-2}$};
\node[right] at (9,8) {$\prn{1+\varepsilon}^{u-1}$};
\node[right] at (9,9) {$\prn{1+\varepsilon}^{u}$};
\end{tikzpicture}
\caption{Reorganization trick: The area below the thick curve is calculated in two different ways.}
\label{fig:reorganization}
\end{figure}

By using the reorganization trick we have, by \Cref{eq:d-def},
\begin{align}
d&=\prn{(1+\eps)^u-(1+\eps)^{u-1}}(n_u-n_{u-1})\nonumber\\
&\qquad+\prn{(1+\eps)^{u-1}-(1+\eps)^{u-2}}(n_u-n_{u-2})
+\dots+(1+\eps)n_u.
\label{eqn:DI}
\end{align}
Therefore it suffices to estimate $n_u-n_j$ for every $1\le j\le u-1$.
We know that $n_u=2^{n-\abs{M}}$.
By definition, $t_j:=n_u-n_j$ counts the sets $S$ such that $Y_S\geq(1+\eps)^j$.
Let $Y:=\prod_{\sqbra{n}\setminus M}Y_{\brkts{i}}$ and observe that $Y_S \ge (1+\eps)^j$ if and only if $Y_{\prn{[n]\setminus M} \setminus S}\le Y/(1+\eps)^j$.
Due to this bijection, $t_j$ also counts the number of sets $S$ such that $Y_S\leq Y/(1+\eps)^j$.

For every $j$, if $Y/(1+\eps)^j<1$ we define $t_j:=0$.
Otherwise, we introduce logarithms to reduce the problem of estimating the number of sets $S\subseteq\sqbra{n}$ such that $\max\left(0,W_S-W\right)/{m_{\min}}=Y_S\leq Y/(1+\eps)^j$ to a $\#\KNAPSACK$ instance
\[
\log W_S\leq\log\cprn{m_{\min}\cdot Y/(1+\eps)^j+W}
\]
which can be more commonly written as $\sum_{i\in S\setminus M}\log w_i\leq B$ for $w_i:={p_i}/\prn{1-p_i}$ (by the definition of $W_S$) and $B:=\log\cprn{m_{\min}\cdot Y/(1+\eps)^j+W}$.
(Note that the latter problem can be reformulated as counting the number of sets $S\subseteq\sqbra{n}\setminus M$ such that $\sum_{i\in S}w_i\leq B$.)

Using \Cref{lem:GKM-1} and \Cref{eqn:DI} we can estimate $t_j$ up to a $(1+\eps)$-approximation in deterministic polynomial time, which in turn would give us a $\prn{1+\varepsilon}$-approximation for $d$ and for that matter a $\prn{1+\varepsilon}^2$-approximation for $\dtv(P,\U)$ by \Cref{eq:dtv-approx}.
Finally, we set $\varepsilon := \Omega(\delta/2)$ so that $(1+\varepsilon)^2 \leq (1+\delta)$ in order to get an approximation ratio of $\prn{1+\delta}$.

The running time is polynomial in $n$ and $1/\delta$ because we ran a polynomial-time approximation algorithm for $\#\KNAPSACK$ polynomially many times.

\subsection{Algorithm for \texorpdfstring{\ProdDtv\ Where $Q$ Has $O(1)$ Parameters}{}}

We will now extend to the case where $Q$ has at most $k$ distinct parameters.
Observe that $\U$ can be viewed as having $k=1$ distinct parameters (equal to $1/2$).
Without loss of generality, let $Q=\bigotimes_i \Bern(q_i) = \Bern(a_1)^{z_1}\otimes\cdots\otimes\Bern(a_k)^{z_k}$ such that $z_1+\dots+z_k=n$.
The main result of this section is the following.

\begin{theorem}
\label{prop:constant-q_i's}
There is an FPTAS for $d_\TV(P,Q)$ where $P$ is an arbitrary product distribution and $Q = \bigotimes_i \Bern(q_i) = \Bern(a_1)^{z_1}\otimes\cdots\otimes \Bern(a_k)^{z_k}$ such that $z_1+\dots+z_k=n$.
\end{theorem}

For simplicity of exposition, we will show first the result for the simpler case when $Q=\Bern(a)^{n}$.

\begin{theorem}
\label{prop:Bern-a}
There is an FPTAS for estimating $\dtv(P,Q)$ where $P$ is an arbitrary product distribution and $Q = \Bern(a)^{n}$ for an $0\leq a\leq 1$.
\end{theorem}

Our approach is to reduce this problem to $\#\KNAPSACK$ with fixed Hamming weights.
If $p_i\geq1/2$ for all $i$, then the latter problem admits an FPTAS due to \Cref{thm:GKM}.
In the scenario where there is an $i$ such that $p_i<1/2$ (in this case the respective $\KNAPSACK$ weight $w_i$ is negative; see our discussion below), we can switch $0$ and $1$ in such coordinates to obtain $\Bern(1-p_i)$ and $\Bern(1-a)$, respectively.
This transformation does not change the distance. 
We show that such instances can be reduced to $\#\KNAPSACK$ with two fixed Hamming weights.

\begin{proof}[Proof of \Cref{prop:Bern-a}]
Let $M$ be the set of indices $i\in\sqbra{n}$ such that $p_i=1$.
We have that $\dtv(P,Q)$ is equal to $\sum_x\max\left(0,P\cprn{x}-Q\cprn{x}\right)$ or
\begin{align*}
&\sum_{S\subseteq [n]:M\subseteq S} \max\left(0,\prod_{i\in S} p_i \prod_{i\notin S} (1-p_i)-a^{|S|}(1-a)^{n-|S|}\right)\\
&\qquad=\prod_{i\notin M} (1-p_i) \sum_{S\subseteq [n]:M\subseteq S} \max\!\left(0,\prod_{i\in S\setminus M} \left(\frac{p_i}{1-p_i}\right)-\frac{1}{\prod_{i\notin M} (1-p_i)} (1-a)^n \left(\frac{a}{1-a}\right)^{|S|}\right)\\
&\qquad=A \sum_{S\subseteq [n]:M\subseteq S}\max\!\left(0,\prod_{i\in S\setminus M} w_i -B\left(\frac{a}{1-a}\right)^{|S|}\right)
\end{align*}
for $A:=\prod_{i\notin M} (1-p_i)$, $w_i:=p_i/\prn{1-p_i}$, and $B:=\prn{1-a}^n/\prod_{i\notin M}\prn{1-p_i}$.

An argument similar to that of \Cref{thm:pqhalf} (based again on the fact that we use finite precision) can be used to show that a normalized version of $\dtv\cprn{P,Q}$ lies in some interval $\sqprn{1,V}$ for $V\leq2^{\poly{n}}$ which again we perceive as $\sqprn{1,V}=\bigcup_{i=1}^{u}[(1+\eps)^i,(1+\eps)^{i+1})$ for $u\leq\cpoly{n}$.
This enables us to use the same approach as in the proof of \Cref{thm:pqhalf}.
Specifically, we approximate $\dtv\cprn{P,Q}$ as $A\cdot m_{\min}\cdot d$ where $d$ is defined as in \Cref{eq:d-def}.
We then approximate $d$ as in the proof of \Cref{thm:pqhalf}, with a notable difference being that now we have to use \Cref{thm:GKM} instead of \Cref{lem:GKM-1} for the $\#\KNAPSACK$ instances to which we reduce the estimation of $\dtv\cprn{P,Q}$.

Therefore, following \Cref{thm:pqhalf}, it would suffice to estimate $d$.
According to \Cref{eqn:DI}, we shall approximate the quantities $t_j:=n_{u}-n_j$ ($n_i$'s as in the proof of \Cref{thm:pqhalf}), which here count the sets $S\subseteq\sqbra{n}$ such that
\begin{align}
\prod_{i\in S\setminus M}w_i
\le B\left(\frac{a}{1-a}\right)^{|S|}+C
=:D
\label{eq:knapsack}
\end{align}
for $C=C\cprn{j}=m_{\min}\cdot Y/(1+\eps)^j$ and the corresponding values of $m_{\min}$ and $Y$ (see the proof of \Cref{thm:pqhalf} for definitions).
Notice how the cardinality $\abs{S}$ of $S$ comes up in the RHS of \Cref{eq:knapsack}.
Since this quantity is not known beforehand, we shall consider cases $\abs{S}=1,\dots,n$ in the $\#\KNAPSACK$ instances that we will solve.
This is the reason we use \Cref{thm:GKM} instead of \Cref{lem:GKM-1}.

First assume that $w_i\ge 1$ for every $i$ (meaning that $p_i\geq1/2$ or $\log w_i\geq0$ for all $i$); we take logarithms (as in \Cref{thm:pqhalf}) to reduce this to a $\#\KNAPSACK$ instance (i.e., $\sum_{i\in S\setminus M}\log w_i\leq\log D$) for every fixed $|S|=1,\dots,n$.
The latter problems can be then solved by the algorithm of \Cref{thm:GKM} for $k=1$ (in the notation of \Cref{thm:GKM}).
Finally, we take the sum of all these counts over the possible values of $\abs{S}$ as our estimate of $t_j$.
Then our estimate for $d$ will come from \Cref{eqn:DI} for $n_u=2^{n-|M|}$.

Now, if for some $i$ we have $w_i<1$ (meaning that $p_i<1/2$ or $\log w_i<0$ for some $i$), then we switch $0$ and $1$ in those coordinates to get $\Bern(1-p_i)$ and $\Bern(1-a)$, respectively.
Then, in $Q$, the first $z$ coin biases are $a$ and the last $n-z$ coin biases are $1-a$ without loss of generality.
In that case, as before, $\dtv(P,Q)$ is
\[
A \sum_{S\subseteq [n]:M\subseteq S}
\max\left(0, \prod_{i\in S\setminus M} w_i -B \prod_{i\in S} v_i\right),
\]
where $w_i\geq1$ for every $i$ and $v_i:=\frac{q_i}{1-q_i}$ whereby $q_i$ is equal to $a$ or $1-a$ depending on whether $w_i$ was originally $\geq1$ or $<1$, respectively.

In this case, for every $S$, its Hamming weight (if we identify a set $S\subseteq\sqbra{n}$ with its characteristic vector in $\bool^n$) in its first $z$ coordinates is $s_1$ and in its last $n-z$ coordinates is $s_2$.
Therefore, it suffices to solve a $\#\KNAPSACK$ instance whereby the quantity $\prod_{i\in S\setminus M} w_i-C$ is at most
\[
B\prn{\frac{a}{1-a}}^{s_1}
\!\prn{\frac{1-a}{a}}^{z-s_1}
\!\prn{\frac{1-a}{a}}^{s_2}
\!\prn{\frac{a}{1-a}}^{n-z-s_2}
\]
for $C=C\cprn{j}=m_{\min}\cdot Y/(1+\eps)^j$ as before (see the proof of \Cref{thm:pqhalf}).
Note that \Cref{thm:GKM} gives an algorithm for the above $\#\KNAPSACK$ problem as well.

We then sum over the counts corresponding to all possible disjoint possibilities of $s_1$ and $s_2$ such that $s_1+s_2=\abs{S}$, for all possible values of $\abs{S}$, to get our estimate of $t_j$.
Then, as earlier, our estimate for $d$ will come from \Cref{eqn:DI} for $n_u=2^{n-|M|}$.
\end{proof}

We now turn to \Cref{prop:constant-q_i's}.

\begin{proof}[Proof of \texorpdfstring{\Cref{prop:constant-q_i's}}{}]
Let $M$ be the set of indices $i\in\sqbra{n}$ such that $p_i=1$.
First, assume that $p_i\geq1/2$ for all $i$.
We have that $\dtv(P,Q)$ is
\begin{align*}
\sum_{S\subseteq [n]:M\subseteq S}
\max\!\bigg(0,\prod_{i\in S} p_i \prod_{i\notin S} (1-p_i)
-a_1^{z_{11}}(1-a_1)^{z_{10}}
\cdots a_k^{z_{k1}}(1-a_k)^{z_{k0}}\bigg),
\end{align*}
where $z_i=z_{i0}+z_{i1}$ and $z_{i0}$ and $z_{i1}$ denote the counts of $0$'s and $1$'s respectively in the characteristic vector of $S$ that correspond to the $a_i$ parameter.

Continuing our manipulation, we see that $\dtv(P,Q)$ is equal to
\begin{multline*}
\prod_{i\notin M} (1-p_i)
\sum_{S\subseteq [n]:M\subseteq S}
\max\!\left(0,\prod_{i\in S\setminus M} \left(\frac{p_i}{1-p_i}\right)
-\frac{a_1^{z_{11}}(1-a_1)^{z_{10}}
\cdots a_k^{z_{k1}}(1-a_k)^{z_{k0}}}{\prod_{i\notin M} (1-p_i)}\right)\\
=A\sum_{S\subseteq\sqbra{n}:M\subseteq S}\max\!\left(0,\prod_{i\in S\setminus M}w_i
-B\cdot\prod_{i=1}^ka_i^{z_{i1}}\prn{1-a_i}^{z_{i0}}\right)
\end{multline*}
for $A:=\prod_{i\notin M} (1-p_i)$, $w_i:=p_i/\prn{1-p_i}$, and $B:=1/\prod_{i\notin M}\prn{1-p_i}$.

An argument similar to that of \Cref{thm:pqhalf} (based again on the fact that we use finite precision) can be used to show that a normalized version of $\dtv\cprn{P,Q}$ lies in some interval $\sqprn{1,V}$ for $V\leq2^{\poly{n}}$ which again we perceive as $\sqprn{1,V}=\bigcup_{i=1}^{u}[(1+\eps)^i,(1+\eps)^{i+1})$ for $u\leq\cpoly{n}$.
This enables us to use the same approach as in the proof of \Cref{thm:pqhalf}.
Specifically, we approximate $\dtv\cprn{P,Q}$ as $A\cdot m_{\min}\cdot d$ where $d$ is defined as in \Cref{eq:d-def}.
We then approximate $d$ as in the proof of \Cref{thm:pqhalf}, with a notable difference being that now we have to use \Cref{thm:GKM} instead of \Cref{lem:GKM-1} for the $\#\KNAPSACK$ instances to which we reduce the estimation of $\dtv\cprn{P,Q}$.

Therefore, following \Cref{thm:pqhalf}, it would suffice to estimate $d$.
According to \Cref{eqn:DI}, we shall approximate the quantities $t_j:=n_{u}-n_j$ ($n_i$'s as in the proof of \Cref{thm:pqhalf}), which here count the sets $S\subseteq\sqbra{n}$ such that
\begin{align}
\prod_{i\in S\setminus M}w_i
\le B\prod_{i=1}^ka_i^{z_{i1}}\prn{1-a_i}^{z_{i0}}+C
=:D
\label{eq:knapsack-2}
\end{align}
for $C=C\cprn{j}=m_{\min}\cdot Y/(1+\eps)^j$, for the corresponding values of $m_{\min}$ and $Y$ (see the proof of \Cref{thm:pqhalf}).

We perform this counting as follows.
We partition the $2^n$ values of $S\subseteq\sqbra{n}$ into subsets corresponding to every possibility of $z_{i1}$'s and $z_{i0}$'s between $0$ and $z_i$, so that $\abs{S}=\sum_{i=1}^kz_{i1}$.
Hence there are at most
\(
\prod_{i=1}^k\prn{z_i+1}\leq\prn{n+1}^k
\)
many parts.
For each such part, we solve a $\#\KNAPSACK$ instance with the following constraints:
\begin{itemize}
\item[$(a)$] Each $z_{i1}$ and $z_{i0}$ correspond to a fixed possibility determined by $S$.
\item[$(b)$] It is the case that $\prod_{i\in S\setminus M} w_i\le D$ for some $D$ determined by $j$ and the $z_{i1}$'s and $z_{i0}$'s as in \Cref{eq:knapsack-2}.
\end{itemize}
Therefore for each part the number of corresponding $\KNAPSACK$ solutions can be approximately counted in polynomial time by the algorithm of \Cref{thm:GKM}.
Our estimate for $t_j$ then is the sum of all of these estimates, which will still be $(1+\varepsilon)$-approximate.
Then (as in \Cref{thm:pqhalf} and \Cref{prop:Bern-a}) our final estimate for $d$ will come from \Cref{eqn:DI} for $n_u=2^{n-|M|}$.

Now, if there is any $p_i<1/2$, then we work with $(1-p_i)$ and $(1-q_i)$ at that particular coordinate and repeat the argument outlined above; this effectively doubles the number of parameters to $2k$ and the resulting algorithm would still run in polynomial time for the case where $k=O(1)$.
\end{proof}

\section{Conclusion}

\label{sec:conclusion}

We initiated a systematic study of the computational nature of the TV distance, a widely used notion of distance between probability distributions.
Our findings are twofold:
On the one hand, we establish hardness results for exactly computing (or approximating) the TV distance (\Cref{thm:hardness-Bern-products-intro}; \Cref{thm:dTV-zero-or-not-NP-complete-intro}).
On the other hand, we present efficient deterministic approximation algorithms (\Cref{thm:pqhalf}; \Cref{prop:constant-q_i's}; \Cref{prop:Bern-a}) for its estimation in some special cases of product distributions.

To conclude, the main open questions that arise from our work are:
\begin{enumerate}
\item
Does there exist an FPTAS for approximating the TV distance between two product distributions?
\item
For what other classes of probabilistic models do there exist TV distance approximation schemes?
\item
What about other notions of distance or similarity between probabilistic models?
\end{enumerate}

\section*{Acknowledgements}

The work of AB was supported in part by National Research Foundation Singapore under its NRF Fellowship Programme (NRF-NRFFAI-2019-0002) and an Amazon Faculty Research Award.
The work of SG was supported by an initiation grant from IIT Kanpur and a SERB award CRG/2022/007985.
Pavan's work is partly supported by NSF award 2130536 and part of the work was done while visiting Simons Institute for the Theory of Computing.
Vinod's work is partly supported by NSF award 2130608 and part of the work was done while visiting Simons Institute for the Theory of Computing.
This work was supported in part by National Research Foundation Singapore under its NRF Fellowship Programme [NRF-NRFFAI1-2019-0004] and an Amazon Research Award.
The work was done in part while AB and DM were visiting the Simons Institute for the Theory of Computing.

\newcommand{\etalchar}[1]{$^{#1}$}

\appendix

\section{Proof of \texorpdfstring{\Cref{thm:GKM}}{}}

\label{sec:appendix}

The proof of \Cref{thm:GKM} follows by adapting the work of \cite{GKM10}.
We first fix some notation and terminology.

A \emph{$\prn{W,n}$-branching program} is a branching program of width $W$ over $n$ Boolean input variables.
A \emph{read-once branching program (ROBP)} is a branching program whereby each input variable is accessed only once.
A \emph{monotone $(W, n)$-ROBP} is a $(W,n)$-ROBP such that in each of its layers $L$ the nodes of $L$ are totally ordered under some relation $\prec$, and whenever $u\prec v$ for some nodes $u$ and $v$ it is the case that the set of partial accepting paths that start $u$ are a subset of the the set of partial accepting paths that start at $v$.

Given a branching program $M$ and a string $z$, the notation $M\cprn{z}$ denotes the output (``accept''/``reject'') of $M$ on input $z$.

An \emph{implicit description} of a monotone ROBP is a description according to which one can efficiently check the relative order of two nodes under $\prec$ (within any layer), and given a node $u$ one can efficiently compute its neighbors.

The following notion of small-space sources was introduced by~\cite{KRVZ06}.

\begin{definition}[\cite{KRVZ06}]
\label{def:small-space}
A \emph{width-$w$ small-space source} is described by a $\prn{w,n}$-branching program $D$ with an additional probability distribution $p_v$ on the outgoing edges associated with vertices $v\in D$.
Samples from the source are generated by taking a random walk on $D$ according to the $p_v$'s and outputting the labels of the edges traversed.
\end{definition}

We require the following useful lemmas from \cite{GKM10}.

\begin{lemma}[\cite{GKM10}]
\label{lem:GKM-3}
Given a ROBP $M$ of width at most $W$ and a small-space source $D$ of width at most $S$, $\cprq{M(x)=1}{x\sim D}$ can be computed exactly in time $O(nSW)$.
\end{lemma}

\begin{lemma}[\cite{GKM10}]
\label{lem:GKM-4}
Given a $(W, n)$-ROBP $M$, the uniform distribution over $M$’s accepting inputs, $\brkts{x \mid M(x) = 1}$ is a width $W$ small-space source.
\end{lemma}

We further require the following result from \cite{GKM10}.

\begin{lemma}[\cite{GKM10}]
\label{thm:Thm-4.2}
Given a monotone $(W, n)$-ROBP $M$, $\delta > 0$, and a small-space source $D$ over $\bool^n$ of width at most $S$, there exists an $(O(n^2S/\delta),n)$-monotone ROBP $M_0$ such that
for all $z$, $M(z) \leq M_0(z)$ and
\[
\cprq{M(z)=1}{z\sim D}
\leq\cprq{M_0(z)=1}{z\sim D}
\leq(1+\delta)\cprq{M(z)=1}{z\sim D}.
\]
Moreover, given an implicit description of $M$ and a description of $D$, $M_0$ can be constructed in deterministic time $O(n^3S(S+\log W)\log(n/\delta)/\delta)$.
\end{lemma}

The main take-away of \Cref{thm:Thm-4.2} is that the number of accepting paths of $M_0$ (under the distribution $D$) approximates the number of accepting paths of $M$ (under the distribution $D$), and moreover $M_0$ has small width.

We now turn to the proof of \Cref{thm:GKM}.

\begin{proof}[Proof of \Cref{thm:GKM}]
Take $M$ in \Cref{thm:Thm-4.2} to be a ROBP for $\KNAPSACK$.
In particular, $M$ decides the validity of the inequality $\sum_{i\in S}a_i\leq b$, which may be also written as $\sum_{i\in\sqbra{n}}a_ix_i\leq b$ if we let $x_i=1$ if and only if $i\in S$.
The ROBP $M$ has $n+1$ layers; layer $0$ has a single start node.
Every other layer $i$ has a node for each partial sum $\sum_{j\leq i}a_ix_i$.
For a node $v$ in layer $i-1$ and $x_i\in\bool$, the $x_i$-th neighbor of $v$ is $v+a_ix_i$.
Naturally, the nodes in the last layer are either rejecting (if their label is more than $b$) or accepting (otherwise).

Note that $M$ may have width $W$ (at most) exponential in $n$; this makes it prohibitive in terms of running time to directly use \Cref{lem:GKM-3} in order to count $\KNAPSACK$ solutions.
Therefore, \Cref{thm:Thm-4.2} comes handy here.

To apply \Cref{thm:Thm-4.2}, let us first note that $M$ is monotone.
Indeed, we can define a total node ordering $\prec$ within each layer of $M$ as follows:
Given two nodes $u,v$ that both belong to some layer of $M$, we define $u\prec v$ if and only if $u>v$.
This satisfies the requirements of a ROBP being monotone as in this case the partial solutions that start at $u$ are a subset of the partial solutions that start at $v$, since the smaller partial sum $v$ allows for more flexibility with respect to the items that we can add to its associated solution.

So by \Cref{thm:Thm-4.2} we can construct in time $O(n^3S(S+\log W)\log(n/\delta)/\delta)$ some ROBP $M_0$ which has width $W_0=O(n^2S/\delta)$ and is such that the probability that $M_0$ is accepting under the distribution $D$ approximates the probability that $M$ is accepting under the distribution $D$.

By \Cref{lem:GKM-4}, the Hamming weight constraints of \Cref{thm:GKM} can be sampled by some small space source of width at most $S\leq\cpoly{n}$, since there is some ROBP of width $\prod_{i=1}^k\prn{\abs{S_i}+1}\leq\prn{n+1}^k=\cpoly{n}$ that only accepts the set of strings that satisfy the Hamming weight constraints of \Cref{thm:GKM}.

This means that the width of $M_0$, namely $W_0$, is at most $O\cprn{\cpoly{n}/\delta}$, and that $M_0$ can be constructed in time $O\cprn{\cpoly{n}/\delta}$ (since $\log W=O\cprn{\cpoly{n}}$).
By \Cref{lem:GKM-3}, we can compute the probability that $M_0$ is accepting under the distribution $D$ in time $O\cprn{nSW_0}=O\cprn{\cpoly{n}/\delta}$.
Let $p$ denote this probability.
As a last step, we multiply $p$ by
\[
\prod_{i=1}^k
\binom{\abs{S_i}}{r_i},
\]
which is the number of strings in the support of $D$ (i.e., the set of strings that have non-zero probability to be sampled by $D$), to get the number of accepting paths of $M_0$.

Since $p=\cprq{M_0(z)=1}{z\sim D}$ $\prn{1+\delta}$-approximates $\cprq{M_0(z)=1}{z\sim D}$, we get that the number of accepting paths of $M_0$ $\prn{1+\delta}$-approximates the number of accepting paths of $M$.

The result now follows from the fact that $M$ is a ROBP for $\KNAPSACK$, so the number of accepting paths of $M_0$ $\prn{1+\delta}$-approximates the number of $\KNAPSACK$ solutions.

Finally the running time of this procedure is polynomial in $n,\log W$, and $1/\delta$, which is polynomial in $n$ and $1/\delta$ since the width of $M$ is $\log W=\cpoly{n}$.
\end{proof}

\end{document}